\documentclass[copyright,creativecommons]{eptcs}
\providecommand{\event}{GandALF 2019} 

\usepackage{etex}

\usepackage[utf8]{inputenc}
\usepackage[T1]{fontenc}
\usepackage{anyfontsize}
\usepackage{newtxtext,newtxmath} 

\usepackage{microtype}
\usepackage[british]{babel}
\usepackage{soul}

\usepackage[framemethod=tikz]{mdframed}
\usepackage{wrapfig}

\usepackage{amssymb}
\usepackage{amsmath}

\usepackage{amsfonts}
\usepackage{stmaryrd}

\usepackage{calrsfs}
\DeclareMathAlphabet{\pazocal}{OMS}{zplm}{m}{n}

\usepackage{amsthm}
\newtheorem{lemma}{Lemma}
\newtheorem{definition}{Definition}

\usepackage{tabu,booktabs}
\usepackage{multicol,multirow}
\usepackage[inline,shortlabels]{enumitem}
\usepackage{pdflscape}

\usepackage{tikz}
\usetikzlibrary{calc, patterns, shapes, shapes.misc, intersections, positioning, arrows.meta}
\usetikzlibrary{decorations.pathreplacing, backgrounds}
\usetikzlibrary{hobby}

\usepackage[noline,noend]{algorithm2e}
\AlgoDontDisplayBlockMarkers
\DontPrintSemicolon
\LinesNumbered
\let\oldnl\nl
\newcommand{\nonl}{\renewcommand{\nl}{\let\nl\oldnl}}
\SetAlgoCaptionLayout{normal}
\SetAlgoCaptionSeparator{}
\SetAlCapSkip{.5em}
\makeatletter
\renewcommand{\algocf@makecaption}[2]{%
	\addtolength{\hsize}{1.5\algomargin}%
	\parbox[t]{\hsize}{\algocf@captiontext{#1:}{#2}}%
	\addtolength{\hsize}{-1.5\algomargin}%
}
\makeatother
\SetArgSty{textup}

\SetCommentSty{mycommfont}
\SetKwComment{tcp}{}{}
\SetKwProg{Def}{def}{:}{}
\SetKwIF{If}{ElseIf}{Else}{if}{:}{elif}{else:}{}
\SetKwFor{For}{for}{:}{}
\SetKwFor{ForAll}{forall}{:}{}
\SetKwFor{While}{while}{:}{}
\SetKwIF{Catch}{}{Try}{catch}{:}{}{try:}{}
\SetKwFor{Loop}{loop:}{}{}
\SetKwFor{DoPar}{do in parallel:}{}{}
\SetKw{Break}{break}
\SetKw{Continue}{continue}
\SetKw{Raise}{raise}
\SetKw{False}{false}
\SetKw{True}{true}
\SetKwFunction{Solve}{solve}
\SetKwFunction{SolveSingle}{solve-single}
\SetKwFunction{SolveTwo}{solve-both}
\SetKwFunction{SolvePlayer}{solve-player}
\SetKwFunction{Search}{search}
\SetKwFunction{TangleAttr}{TAttr}
\SetKwFunction{BuildRegion}{build-region}
\SetKwFunction{FPI}{fpi}
\SetKwFunction{Winner}{winner}
\SetKwFunction{Onestep}{onestep}
\SetKwFunction{Prioprom}{prioprom}
\SetKwFunction{AAttr}{attr}
\SetKwFunction{Zielonka}{zielonka}
\SetKwFunction{SearchDominion}{search-dominion}
\SetKwFunction{Search}{search}
\SetKwFunction{TangleLearning}{tangle-learning}

\newcommand{\cupdot}{\mathbin{\mathaccent\cdot\cup}}

\makeatletter
\newcommand{\pto}{}
\newcommand{\pgets}{}

\DeclareRobustCommand{\pto}{\mathrel{\mathpalette\p@to@gets\to}}
\DeclareRobustCommand{\pgets}{\mathrel{\mathpalette\p@to@gets\gets}}

\newcommand{\p@to@gets}[2]{%
	\ooalign{\hidewidth$\m@th#1\mapstochar\mkern5mu$\hidewidth\cr$\m@th#1\to$\cr}%
}
\makeatother

\newcommand{\Attr}{\mathit{Attr}}
\newcommand{\TAttr}{\mathit{TAttr}}
\newcommand{\pr}{\textsf{pr}}
\newcommand{\invpr}{\pr^{-1}}
\newcommand{\measure}{\textsf{r}}
\newcommand{\invmeasure}{\measure^{-1}}
\newcommand{\invalpha}{{\overline{\alpha}}}

\newcommand{\sqdiamond}{\tikz [x=1.2ex,y=1.2ex,line width=.08ex] \draw (0,.5) -- (.5,1) -- (1,.5) -- (.5,0) -- (0,.5) -- cycle;}
\newcommand{\sqsq}{\tikz [x=0.95ex,y=1ex,line width=.1ex] \draw (0,0) -- (1,0) -- (1,1) -- (0,1) -- (0,0) -- cycle;}

\newcommand{\Even}{{\scalebox{0.95}{\sqdiamond}}}
\newcommand{\Odd}{{\scalebox{0.9}{$\sqsq$}}}

\newcommand{\Veven}{V_{\Even}}
\newcommand{\Vodd}{V_{\Odd}}

\title{A Parity Game Tale of Two Counters}
\def\titlerunning{A Parity Game Tale of Two Counters}

\author{
	Tom van Dijk\
	\institute{Formal Methods and Tools \\ University of Twente, Enschede}
	\email{t.vandijk@utwente.nl}
}
\def\authorrunning{Tom van Dijk}

\begin{document}

\maketitle

\begin{abstract}

Parity games are simple infinite games played on finite graphs with a winning condition that is expressive enough to capture nested least and greatest fixpoints.
Through their tight relationship to the modal mu-calculus, they are used in practice for the model-checking and synthesis problems of the mu-calculus and related temporal logics like LTL and CTL.
Solving parity games is a compelling complexity theoretic problem, as the problem lies in the intersection of UP and co-UP and is believed to admit a polynomial-time solution, motivating researchers to either find such a solution or to find superpolynomial lower bounds for existing algorithms to improve the understanding of parity games.

We present a parameterized parity game called the Two Counters game, which provides an exponential lower bound for a wide range of attractor-based parity game solving algorithms.
We are the first to provide an exponential lower bound to priority promotion with the delayed promotion policy, and the first to provide such a lower bound to tangle learning.

\end{abstract}

\section{Introduction}
\label{sec:introduction}

Parity games are turn-based games played on a finite graph.
Two players \emph{Odd} and \emph{Even} play an infinite game by moving a token along the edges of the graph.
Each vertex is labeled with a natural number \emph{priority} and the winner of the game is determined by the parity of the highest priority that is encountered infinitely often.
Player Odd wins if this parity is odd; otherwise, player Even wins.

Parity games play a central role in several domains in theoretical computer science.
Their study has been motivated by their relation to various problems in formal verification and synthesis that can be reduced to solving parity games, as parity games capture the expressive power of nested least and greatest fixpoint operators~\cite{DBLP:conf/cav/Fearnley17}.
Deciding the winner of a parity game is polynomial-time equivalent to checking non-emptiness of non-deterministic parity tree automata~\cite{DBLP:conf/stoc/KupfermanV98}, and to the explicit model-checking problem of the modal $\mu$-calculus~\cite{DBLP:journals/tcs/EmersonJS01,DBLP:conf/dagstuhl/2001automata,DBLP:journals/tcs/Kozen83}.
Synthesis problems ask for an implementation of a system that satisfies the desired properties and reactive synthesis tools like \textsc{Strix} use parity games to synthesize controllers satisfying LTL formulas~\cite{DBLP:conf/cav/MeyerSL18}.
Solving the related parity game either results in a correct implementation or in a counterexample demonstrating how an adversary can make the property fail.

Parity games are also interesting for complexity theory, 
as the problem of determining the winner of a parity game is known to lie in $\text{UP}\cap\text{co-UP}$~\cite{DBLP:journals/ipl/Jurdzinski98},
which is contained in $\text{NP}\cap\text{co-NP}$~\cite{DBLP:journals/tcs/EmersonJS01}.
This problem is therefore unlikely to be NP-complete and it is widely believed that a polynomial solution exists.

Recent work proposes novel algorithms based on the notion of a tangle~\cite{DBLP:conf/cav/Dijk18}.
A tangle is a strongly connected subgame of a parity game for which one player has a strategy to win all cycles in the subgame.
Tangles play a fundamental role in various parity game algorithms, but most algorithms are not explicitly aware of tangles and can explore the same tangles repeatedly~\cite{DBLP:conf/cav/Dijk18}.
The algorithms proposed in~\cite{DBLP:conf/cav/Dijk18} solve parity games by explicitly computing tangles using attractor computation.

Tangles are related to snares~\cite{DBLP:conf/lpar/Fearnley10} and quasi-dominions~\cite{DBLP:conf/cav/BenerecettiDM16}, with the critical difference that tangles are strongly connected whereas snares and quasi-dominions may be composed of multiple tangles.
Thus it is an obvious question whether the subexponential counterexample of Friedmann~\cite{DBLP:journals/dam/Friedmann13} to Fearnley's snare-based algorithm~\cite{DBLP:conf/lpar/Fearnley10}
can be adapted for tangle learning.
We show that this is possible and that the resulting counterexample is powerful enough to be a difficult lower bound to a wide range of algorithms.

We propose a parameterized parity game based on binary counters.
The goal is to trick the algorithms to explore the progression of the counters, only solving the game when all bits are set.
The critical ingredient to make these games difficult for attractor-based algorithms is to use two intertwined binary counters, one for each player, that progress together.
We show empirically that these games are difficult for a wide range of algorithms, in particular for those based on attractor computation, such as priority promotion~\cite{DBLP:conf/cav/BenerecettiDM16,DBLP:journals/fmsd/BenerecettiDM18} and its variations~\cite{DBLP:journals/corr/BenerecettiDM16,DBLP:conf/hvc/BenerecettiDM16}, tangle learning~\cite{DBLP:conf/cav/Dijk18}, and the recursive algorithm by Zielonka~\cite{DBLP:journals/apal/McNaughton93,DBLP:journals/tcs/Zielonka98}.
For these, we provide the exponential lower bound of
$\Omega(2^{\sqrt{n}})$.
We are the first to provide an exponential lower bound for the delayed promotion policy of priority promotion and for tangle learning.

\section{Preliminaries}
\label{sec:preliminaries}


A parity game $\Game$ is a tuple $(\Veven, \Vodd, E, \pr)$ where $V=\Veven\cupdot \Vodd$ is a set of $n$ vertices partitioned into the sets $\Veven$ controlled by player \emph{Even} and $\Vodd$ controlled by player \emph{Odd}, and $E\subseteq V\times V$ is a left-total binary relation describing all moves, i.e., every vertex has at least one successor.
We also write $E(u)$ for all successors of $u$ and $u\rightarrow v$ for $v\in E(u)$.
The function $\pr\colon V\rightarrow \{0,1,\dotsc,d\}$
assigns to each vertex a \emph{priority}, where $d$ is the highest priority in the game.
We write $\alpha\in\{\Even,\Odd\}$ to denote a player $\Even$ or $\Odd$ and $\invalpha$ for the opponent of $\alpha$. 
In visual representations of parity games, we use diamonds and boxes for vertices of the two players. 

We write 
$\pr(V)$ for $\max\{\pr(v)\mid v\in V\}$ and 
similarly $\pr(\Game)$ for the highest priority in $\Game$.
Furthermore, we write $\pr^{-1}(i)$ for all vertices with the priority $i$.
A \emph{play} $\pi=v_0 v_1 \dots$ is an infinite sequence of vertices consistent with $E$, i.e.,
$v_i \rightarrow v_{i+1}$ for all $i\geq 0$. 
We denote with $\inf(\pi)$ the vertices in $\pi$ that occur infinitely many times in $\pi$.
Player Even wins a play $\pi$ if $\pr(\inf(\pi))$ is even; player Odd wins if $\pr(\inf(\pi))$ is odd.
We write $\text{Plays}(v)$ to denote all plays starting at vertex $v$
and $\text{Plays}(V)$ for $\{\pi\in\text{Plays}(v)\mid v\in V\}$.

A (positional) \emph{strategy} $\sigma\colon V\nrightarrow V$ assigns to each vertex in its domain a single successor in $E$, i.e., $\sigma\subseteq E$.
We refer to a strategy of player $\alpha$ to restrict the domain of $\sigma$ to $V_\alpha$. 
In the remainder, all strategies $\sigma$ are of a player $\alpha$.
Given a strategy $\sigma$, we define the subgame induced by $\sigma$ as $\Game[\sigma] := (V_\Even, V_\Odd, E', \pr)$ where $E':=\{ uv \in E \mid u\notin\textup{dom}(\sigma) \lor \sigma(u)=v \}$.
We write $\text{Plays}(v, \sigma)$ for all plays from $v$ consistent with $\sigma$, i.e., $\text{Plays}_{\Game[\sigma]}(v)$, 
and $\text{Plays}(V,\sigma)$ for $\{\pi\in\text{Plays}(v,\sigma)\mid v\in V\}$.

A basic result for parity games is that they are memoryless determined~\cite{DBLP:conf/focs/EmersonJ91}, i.e., each vertex is either winning for player Even or for player Odd, and the players have a positional strategy for their winning vertices.
Player $\alpha$ wins vertex $v$ if there is a strategy $\sigma$ such that player $\alpha$ wins all plays in $\text{Plays}(v,\sigma)$.

A set of vertices $U$ is called \emph{closed} with respect to player $\alpha$ if all vertices that belong to player $\alpha$ have a successor in $U$ and all vertices of player $\invalpha$ only have successors inside $U$.
That is, player $\invalpha$ cannot leave $U$.
If a set of vertices is closed in the full game $\Game$ then we also call the set \emph{globally closed} and if it is closed in some subgame of $\Game$ then we call the set \emph{locally closed} in that subgame.

A \emph{dominion} $D$ is a (globally) closed set of vertices for which player $\alpha$ has a strategy $\sigma$ such that player $\alpha$ wins all plays in $\text{Plays}(D,\sigma)$. 
That is, player $\alpha$ wins all cycles in the subgame $\Game[D,\sigma]$ induced by $D$ and $\sigma$. 
We also write a \emph{$p$-dominion} for a dominion where $p$ is the highest priority encountered infinitely often in plays consistent with $\sigma$, i.e., $p:=\max\{\pr(\inf(\pi))\mid \pi\in\textrm{Plays}(D,\sigma)\}$.

Several algorithms for solving parity games employ \emph{attractor computation}.
Given a set of vertices $A$,
the attractor of $A$ for a player $\alpha$ contains exactly those vertices from which player $\alpha$ can ensure arrival in $A$.
We write $\Attr^\Game_\alpha(A)$ to attract vertices in $\Game$ to $A$ as player $\alpha$,
i.e., the least fixpoint of
\[
Z := A \cup \{ v\in V_\alpha \mid E(v)\cap Z \neq \emptyset \} \cup \{ v\in V_{\invalpha} \mid E(v)\subseteq Z \}
\]
We compute the $\alpha$-attractor of $A$ with a backward search from $A$, initially setting $Z:=A$ and iteratively adding $\alpha$-vertices with a successor in $Z$ and $\invalpha$-vertices with no successors outside $Z$.
We call a set of vertices $A$ $\alpha$-maximal if $A=\Attr^\Game_\alpha(A)$.
The attractor also yields an ``attractor strategy'' by selecting a vertex in $Z$ for every added $\alpha$-vertex $v$ when $v$ is added to $Z$, and by selecting a vertex in $Z$ for all $\alpha$-vertices in $A$ that do not yet have a strategy but can play to $Z$.

Attractors are often used to attract to a set $A:=\pr^{-1}(\pr(\Game))$ for the player that wins the highest priority in $\Game$.
By repeatedly computing this attractor and removing it from the game, the game is decomposed into \emph{regions} each associated with the priority $p$ and player $\alpha=p \bmod 2$.
Each such region has the property that all plays that stay in the region are won by player $\alpha$, and that player $\invalpha$ can leave the region to higher $\alpha$-regions and, only via a vertex of priority $p$, to lower regions.
The recursive algorithm~\cite{DBLP:journals/tcs/Zielonka98} investigates whether (lower) $\alpha$-regions attract vertices from higher $\invalpha$-regions.
Priority promotion~\cite{DBLP:conf/cav/BenerecettiDM16} merges locally closed regions with higher $\alpha$-regions that the opponent can escape to.
Tangle learning~\cite{DBLP:conf/cav/Dijk18} computes tangles from locally closed regions which are used to improve the attractor-based decomposition. 

\section{Recursive Algorithm}
\label{sec:zielonka}

\begin{algorithm}[tbp]
\Def{\Zielonka{$\Game$}}{
	\lIf(\tcp*[f]{solve the empty game}){$\Game = \emptyset$}{\Return $\emptyset, \emptyset$}
	$p \leftarrow \pr(\Game),\ \alpha \leftarrow \pr(\Game) \bmod 2$ \; 
	$A$ $\leftarrow$ $\smash{\text{Attr}^{\Game}_{\alpha}(\pr^{-1}(p))}$ \tcp*{attract to highest priority}
	$W_\Even, W_\Odd$ $\leftarrow$ \Zielonka{$\Game\setminus A$} \tcp*{compute remaining subgame}
	$B$ $\leftarrow$ $\smash{\text{Attr}^{\Game}_{\invalpha}(W_{\invalpha})}$ \tcp*{attract to opponent region}
	\If{$B = W_{\invalpha}$}{
		$W_{\alpha}$ $\leftarrow$ $W_\alpha \cup A$ \tcp*{$A$ is won by $\alpha$}
	}
	\Else{
		$W_\Even, W_\Odd$ $\leftarrow$ \Zielonka{$\Game\setminus B$} \tcp*{recompute remaining subgame}
		$W_{\invalpha}$ $\leftarrow$ $W_{\invalpha} \cup B$ \tcp*{$B$ is won by $\invalpha$}
	}
	\Return {$W_{\Even}, W_{\Odd}$} \;
}
\caption{Zielonka's recursive algorithm.}
\label{alg:zielonka}
\end{algorithm}



The recursive algorithm by Zielonka~\cite{DBLP:journals/tcs/Zielonka98} is a well-known algorithm to solve parity games with fast practical performance~\cite{DBLP:conf/tacas/Dijk18,DBLP:conf/atva/FriedmannL09}.
Although it requires exponential time in the worst-case, tight bounds are known for various classes of games~\cite{DBLP:journals/corr/GazdaW13}.
Very recently, Parys has proposed a version of the recursive algorithm that runs in quasi-polynomial time~\cite{DBLP:journals/corr/abs-1904-12446}.

The recursive algorithm is based on attractor computation.
See Algorithm~\ref{alg:zielonka}.
We omit computing the winning strategy.
At each step, given the current subgame $\Game$, the algorithm computes the attractor $A$ to the highest vertices in $\Game$ and recursively solves the subgame $\Game\setminus A$. 
If the opponent $\invalpha$ can attract vertices in $A$ to $W_\invalpha$, then $\invalpha$ wins $B := \text{Attr}^\Game_\invalpha(W_\invalpha)$.
The vertices attracted to $W_\invalpha$ are vertices of priority $p$ that are now estimated to be won by player $\invalpha$ plus vertices of priority $<p$ that player $\invalpha$ attracts to $W_\invalpha$ via vertices of priority $p$.
The algorithm then computes the remaining subgame $\Game\setminus B$ recursively and returns the winning regions.
If player $\invalpha$ does not attract any vertices to $W_\invalpha$, then no second recursion is needed.
The winning strategies are trivially obtained as the attractor strategy plus any successor in the set $W_\alpha\cup A$ for $\alpha$-vertices with priority $p$.
One can also view the recursive algorithm as decomposing the game into $\alpha$-maximal regions and refining these regions starting with the lowest region.

\section{Priority Promotion}
\label{sec:prioprom}

\begin{algorithm}[tbp]
\Def{\SearchDominion{$\Game$}}{
	$\measure \leftarrow V \mapsto \bot$ \tcp*{all vertices to $\bot$}
	$p \leftarrow \pr(\Game)$ \tcp*{start at highest priority}
	\While{True}{
		$\alpha \leftarrow p \bmod 2$ \tcp*{current player}
		$\textup{Subgame} \leftarrow V\setminus \{ v \mid \measure(v) > p \}$ \tcp*{current subgame}
		$A \leftarrow \invmeasure(p) \cup (\invpr(p) \cap \textup{Subgame})$ \tcp*{current region/attractor target}
		$Z \leftarrow \smash{\text{Attr}^{\Game\cap \textup{Subgame}}_{\alpha}(A)}$ \tcp*{attract to $A$}
		$\textup{Open} \leftarrow \{ v \in Z \cap V_{\alpha} \mid E(v) \cap Z = \emptyset \}$ \tcp*{open $\alpha$-vertices}
		$\textup{Esc} \leftarrow E(Z \cap V_{\invalpha})\setminus Z$ \tcp*{escape targets for $\invalpha$}
		\If(\tcp*[f]{\textcolor{black}{\underline{is the region open}}?}){$\textup{Open}\neq\emptyset \vee (\textup{Esc}\cap \textup{Subgame})\neq\emptyset$}{
			$\measure \leftarrow \measure[Z \mapsto p]$ \tcp*{record the region}
			$p \leftarrow \pr(\textup{Subgame}\setminus Z)$ \tcp*{continue with next highest priority}
		}
		\ElseIf(\tcp*[f]{\textcolor{black}{\underline{locally closed}}?}){$\textup{Esc}\neq\emptyset$}{
			$p \leftarrow \min\{ \measure(v) \mid v \in \textup{Esc} \}$ \tcp*{set $p$ to lowest escape}
			$\measure \leftarrow \measure[Z \mapsto p][\{v\mid \measure(v)<p\} \mapsto \bot]$ \tcp*{merge, reset lower than $p$}
		}
		\Else(\tcp*[f]{\textcolor{black}{\underline{globally closed}}?}){
			$Z \leftarrow \Attr^\Game_\alpha(Z)$ \tcp*{maximize the dominion}
			\Return $\alpha,Z$ \tcp*{dominion of player $\alpha$!}
		}
	}
}
\BlankLine
\Def{\Prioprom{$\Game$}}{
	$W_\Even\leftarrow\emptyset$, $W_\Odd\leftarrow\emptyset$ \tcp*{initialize sets}
	\While{$\Game\neq\emptyset$}{
		$\alpha, D$ $\leftarrow$ \SearchDominion{$\Game$} \tcp*{compute the next dominion}
		$W_\alpha \leftarrow W_\alpha \cup D$ \tcp*{add dominion to winning region of player $\alpha$}
		$\Game \leftarrow \Game\setminus D$ \tcp*{remove dominion from game}
	}
	\Return $W_\Even, W_\Odd$
}
\caption{The priority promotion algorithm.}
\label{alg:prioprom}
\end{algorithm}


Priority promotion was proposed in \cite{DBLP:conf/cav/BenerecettiDM16,DBLP:journals/fmsd/BenerecettiDM18} and improved in~\cite{DBLP:journals/corr/BenerecettiDM16,DBLP:conf/hvc/BenerecettiDM16}.
Like the recursive algorithm, priority promotion computes the top-down $\alpha$-maximal decomposition of the game into regions.
These regions have the property that all plays that stay in the region are won by player $\alpha$.
Regions are \emph{locally closed} when all vertices of player $\alpha$ have a successor in the region and no vertices of player $\invalpha$ can play to lower regions.
Because all regions are $\alpha$-maximal, player $\invalpha$ can only escape to ``higher'' regions of player $\alpha$.
Locally closed regions are merged with the lowest higher region to which player $\invalpha$ can escape, after which the decomposition of the game is refined by attracting to this merged region and recomputing the lower regions.
This is called promoting, as the lower region is ``promoted'' to the higher region.

See Algorithm~\ref{alg:prioprom}.
Again we omit explicitly computing the winning strategies, but they are computed trivially by the attractors.
The \texttt{search-dominion} algorithm is given a game and returns a dominion and the winner of this dominion.
By repeatedly calling \texttt{search-dominion} and removing the dominion from $\Game$, priority promotion solves parity games.
The \texttt{search-dominion} algorithm refines the regions until a dominion is found.
The function $\textsf{r}$ records the current region (priority $p$) of each vertex, or $\bot$ if the vertex is not in a region.
After promoting a region to a higher priority $p$ (lines~15--16) the algorithm resets all regions below $p$ and then attracts to the merged region at lines~5--8.

Consider some region of priority $4$ that is promoted to priority $16$.
By attracting to the combined region $16$, priority promotion attracts vertices from regions $5\dots15$ of player $\invalpha$ that are attracted to the lower region $4$ in one step, while the recursive algorithm attracts from each affected region of $\invalpha$ separately.
Another difference with the recursive algorithm is that priority promotion always promotes the highest locally closed region whereas the recursive algorithm starts refining with the lowest regions.

The original algorithm~\cite{DBLP:conf/cav/BenerecettiDM16} resets all lower regions after a promotion.
These lower regions might be the result of earlier promotions and these promotions are then often repeated.
The PP+ extension~\cite{DBLP:journals/corr/BenerecettiDM16} only resets lower regions of player $\invalpha$, since promoting regions of player $\alpha$ can only ``break'' regions of player $\invalpha$.
The region restoration (RR) extension~\cite{DBLP:conf/hvc/BenerecettiDM16} further improves upon PP+ by only resetting a lower region if the earlier attractor strategy of the player for the vertices that are still in the original region now leaves this region to a higher region (of the opponent).
The delayed promotion (DP) policy (also~\cite{DBLP:journals/corr/BenerecettiDM16}) uses a heuristic to delay promotions.
The algorithm records which regions are the result of a promotion and 
if another promotion might affect a merged region, then this promotion is delayed.
The delayed promotions are instead performed on a copy of the decomposition.
When no more normal promotions can be performed, the delayed promotions of one player are applied from this copy, where this player is the one that wins the highest merged region in the copy.
The other delayed promotions are discarded and the algorithm continues with an empty record of merged regions.

\section{Tangles and Tangle Learning}
\label{sec:tangles}

Earlier work introduced tangles as substructures of parity games~\cite{DBLP:conf/cav/Dijk18}.
Tangles are strongly connected subgames of a parity game for which one player has a strategy to win all cycles in the subgame.
The losing player must therefore escape the tangle, so we extend attractor computation to simultaneously attract all vertices in a tangle when the losing player must escape to the attracting set.
This leads to the \emph{tangle learning} algorithm, which computes new tangles along the decomposition of the game, computed using the extended attractor.
The algorithm solves parity games by finding tangles that are dominions.

\begin{definition}
	A \emph{tangle} is a pair $T=(U,\sigma)$ where $U$ is a nonempty set of vertices $U\subseteq V$ and $\sigma\colon U_\alpha\to U$ is a strategy for all vertices in $U$ of player $\alpha := \pr(U)\bmod 2$, such that the subgame $\Game[U,\sigma]$ induced by the tangle
	is strongly connected and player $\alpha$ wins all cycles in $\Game[U,\sigma]$.
\end{definition}
We say that a tangle with the highest priority $p$ is a $p$-tangle, and that it is an $\alpha$-tangle. 
We often use $T$ for just the set of vertices in the tangle, e.g., 
writing $v\in T$ if a vertex $v$ is in the tangle.
We write $E_T(T)$ for the successors outside $T$ from $\invalpha$-vertices in $T$,
$E_T(T) := \{\, v \in V\setminus T \mid u \to v \mid u\in T\cap V_\invalpha \,\}$.

We have several basic observations related to tangles~\cite{DBLP:conf/cav/Dijk18}.
A \emph{closed} $p$-tangle, from which player $\invalpha$ cannot leave, is a $p$-dominion.
Vice versa, every $p$-dominion contains at least one $p$-tangle.
Furthermore, tangles are often composed of subtangles with lower priorities.
We thus find a hierarchy of tangles in any dominion $D$ with winning strategy $\sigma$
by computing the set of winning priorities
$\{\,\pr(\inf(\pi))\mid \pi\in\textrm{Plays}(D,\sigma)\,\}$.
There is a $p$-tangle in $D$ for every $p$ in this set.
Tangles are thus intuitive substructures of dominions.
One can find all tangles in a dominion $D$ by computing $\{\,\inf(\pi)\mid \pi\in\textrm{Plays}(D,\sigma)\,\}$.
See for example Figure~\ref{fig:tangled}. Player Odd wins 
with 
strategy $\{\,\textbf{d}\rightarrow\textbf{e}\,\}$.
Player Even can avoid priority $5$, but then loses with priority $3$.
The 5-dominion contains a $5$-tangle and a $3$-tangle.

\begin{figure}[b]
	\vspace{-1em}
	\centering
	\scalebox{0.9}{
		\begin{tikzpicture}		
		\tikzset{every edge/.append style={>=stealth,->,solid,thick,draw,text height=0.5ex,text depth=0.2ex}}
		\tikzset{every node/.append style={minimum size=6mm,draw,fill=black!10}}
		\tikzset{my label/.style args={#1:#2}{
				append after command={
					($(\tikzlastnode.center)$) coordinate [label={[label distance=5mm,black]#1:\textbf{\strut #2}}]
		}}}
		\tikzstyle{even}=[diamond,minimum size=6mm,draw]
		\tikzstyle{odd}=[regular polygon,regular polygon sides=4,minimum size=6mm,draw]
		\tikzstyle{seven}=[even,fill=white]
		\tikzstyle{sodd}=[odd,fill=white]
		
		\draw[pattern=crosshatch,pattern color=blue!15] plot [smooth cycle, tension=0.5]
		coordinates {(-0.5,1.85) (-0.5,-0.4) (2,-0.45) (2,1.9)};
		\draw[fill,white] plot [smooth cycle, tension=0.5]
		coordinates {(0.85,1.8) (0.8,-0.25) (1.95,-0.35) (2,1.8)};		
		\draw[pattern=crosshatch dots,pattern color=blue!35] plot [smooth cycle, tension=0.5]
		coordinates {(0.85,1.8) (0.8,-0.25) (1.95,-0.35) (2,1.8)};		
		\draw (0,1.5)     node[seven, my label={above:b}] (a) {5};
		\draw (-1.5,0.75) node[even, my label={above:a}] (b) {6};
		\draw (0,0)       node[odd, my label={below:d}]  (c) {1};
		\draw (1.5,0)     node[seven, my label={below:e}] (d) {3};
		\draw (1.5,1.5)   node[even, my label={above:c}] (e) {2};
		\draw (a) edge (c);
		\draw (b) edge (a);
		\draw (c) edge (d);
		\draw (c) edge (b);
		\draw (d) edge [bend left=20] (e);
		\draw (e) edge [bend left=20] (d);
		\draw (e) edge (a);
		\end{tikzpicture}
	}
	\vspace{-1em}
	\caption{A $5$-dominion with a $5$-tangle and a $3$-tangle}
	\label{fig:tangled}
\end{figure}
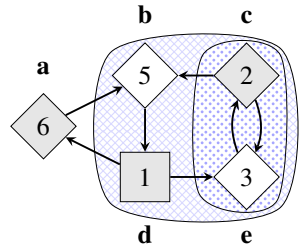 

As described in~\cite{DBLP:conf/cav/Dijk18}, tangles play a central role for various parity game solving algorithms, as they implicitly explore tangles and may explore the same tangles repeatedly, especially when tangles are nested.
This motivates algorithms that explicitly target tangles, such as the ``snare memoization'' extension to strategy improvement~\cite{DBLP:conf/lpar/Fearnley10} and the ``tangle attractor'' approach of tangle learning~\cite{DBLP:conf/cav/Dijk18}.

\begin{algorithm}[tbp]
\Def{\Search{$\Game$, $\textbf{\textup{T}}$}}{
	\lIf(\tcp*[f]{no tangles in an empty game}){$\Game=\emptyset$}{\Return $\emptyset$}
	$p \leftarrow \pr(\Game)$, $\alpha \leftarrow \pr(\Game) \bmod 2$ \tcp*{obtain current priority and player}
	$Z,\sigma \leftarrow \TAttr^{\Game,\textbf{\textup{T}}}_\alpha\big(\pr^{-1}(p)\big)$ \tcp*{tangle-attract to highest vertices}
	$O \leftarrow \{ v\in Z_{\alpha} \mid E(v)\cap Z=\emptyset \} \cup \{ v\in Z_{\invalpha} \mid E(v)\not\subseteq Z \}$ \tcp*{compute open vertices}
	$Y \leftarrow \textbf{ if } O=\emptyset \textbf{ then } \texttt{bottom-sccs(}Z,\sigma\texttt{)} \textbf{ else }\emptyset$ \tcp*{compute tangles if closed}
	\Return $Y \cup \texttt{search(}\Game\setminus Z, \textbf{\textup{T}}\texttt{)}$ \tcp*{find tangles in remaining game}
}
\BlankLine
\Def{\TangleLearning{$\Game$}}{
	$W_\Even\leftarrow\emptyset$,
	$W_\Odd\leftarrow\emptyset$,
	$\textbf{\textup{T}}\leftarrow\emptyset$ \tcp*{initialize sets}
	\While{$\Game\neq\emptyset$}{
		$Y$ $\leftarrow$ \Search{$\Game,\textbf{\textup{T}}$} \tcp*{search for tangles}
		$\textbf{\textup{T}}\leftarrow \textbf{\textup{T}} \cup \{\,T\in Y\mid E_T(T)\neq\emptyset\,\}$ \tcp*{add new open tangles to $\textbf{\textup{T}}$}
		$D\leftarrow\{\,T\in Y\mid E_T(T) = \emptyset\,\}$ \tcp*{obtain new dominions}
		\If{$D\neq\emptyset$}{
			$W_\Even \leftarrow {W_\Even\cup \TAttr_\Even^{\Game,\textbf{\textup{T}}}(\bigcup D_\Even)}, W_\Odd \leftarrow {W_\Odd\cup \TAttr_\Odd^{\Game,\textbf{\textup{T}}}(\bigcup D_\Odd)}$ \tcp*{extend dominions}
			$\Game\leftarrow\Game\setminus (W_\Even \cup W_\Odd)$ \tcp*{remaining game}
			$\textbf{\textup{T}}$ $\leftarrow$ $\textbf{\textup{T}}\cap\Game$ \tcp*{remaining tangles}
		}
	}
	\Return $W_\Even, W_\Odd$
}
\caption{The tangle learning algorithm.}
\label{alg:tanglelearning}
\end{algorithm}

Tangle learning is based on the tangle attractor,
extending
attractor computation to 
attract vertices of $\alpha$-tangles where player $\invalpha$ must play to the attracting set, writing $\TAttr^{\Game, \textbf{T}}_\alpha(A)$ to attract vertices in subgame $\Game$ and vertices of $\alpha$-tangles in the set \textbf{T} that are in subgame $\Game$ to $A$ as player $\alpha$,
i.e., the least fixpoint of
\[
\begin{tabu}{r@{}l}
Z := A &\ \cup\  \{\;v\in V_\alpha \mid E(v)\cap Z \neq \emptyset\;\} \ \cup\ \{\;v\in V_{\invalpha} \mid E(v)\subseteq Z\;\} \\
 & \ \cup\  \{\; v\in T\mid T\in \textbf{\textup{T}}\land T\subseteq V\land \pr(T)\equiv_2\alpha\land
E_T(T)\subseteq Z \;\}
\end{tabu}
\]

See further~\cite{DBLP:conf/cav/Dijk18}.
The tangle learning algorithm (Algorithm~\ref{alg:tanglelearning}) repeatedly decomposes the game with the tangle attractor.
By computing the bottom strongly connected components of the regions, new tangles are obtained and added to the set of tangles \textbf{T} (line~12).
Each iteration of this algorithm adds new tangles to this set, resulting in a different decomposition each time.
Tangles without escapes, i.e., closed tangles, are dominions, which are then maximized using the attractor and removed from the game (lines~13--16).
After removing the dominions, all tangles with solved vertices are removed from \textbf{T} (line~17).
As before, winning strategies are obtained from the attractor.
We omit explicitly computing the winning strategies, but we do need the strategy that the tangle attractor yields for computing the tangles at line~6 as explained in~\cite{DBLP:conf/cav/Dijk18}.
A variation called ``alternating tangle learning'' is also described in~\cite{DBLP:conf/cav/Dijk18}, which alternates between exhaustively computing only tangles for player Odd or for player Even.

A closed region is essentially a collection of possibly unconnected tangles and vertices attracted to these tangles.
Compared to tangle learning, priority promotion ``attracts'' all tangles in the region to the lowest region that would attract a tangle in the region, and then discards knowledge of the tangles.

\section{Distractions}
\label{sec:distractions}

\begin{figure}[t]
	\centering
	\resizebox{0.7\linewidth}{!}{
		\begin{tabu} to 0.8\linewidth {X[lm]X[0.2,cm]X[rm]}
			\begin{tikzpicture}[baseline=(current bounding box.center)]	
			\tikzset{every edge/.append style={>=stealth,->,solid,thick,draw,text height=0.5ex,text depth=0.2ex}}
			\tikzset{every node/.append style={minimum size=6mm,draw,fill=black!10}}
			\tikzset{my label/.style args={#1:#2}{
					append after command={
						($(\tikzlastnode.center)$) coordinate [label={[label distance=5mm,black]#1:\textbf{\strut #2}}]
			}}}
			\tikzstyle{even}=[diamond,minimum size=6mm,draw]
			\tikzstyle{odd}=[regular polygon,regular polygon sides=4,minimum size=6mm,draw]
			\tikzstyle{seven}=[even,fill=white]
			\tikzstyle{sodd}=[odd,fill=white]
			
			\draw[pattern=crosshatch,pattern color=blue!25] plot [smooth cycle, tension=0.5]
			coordinates {(-2.0,2) (-2.0,0.9) (0.5,0.9) (0.5,2)};
			
			\draw[pattern=crosshatch,pattern color=red!25] plot [smooth cycle, tension=0.5]
			coordinates {(-0.5,0.5) (-0.5,-0.5) (1.9,-0.5) (1.9,0.5)};
			
			\draw (-1.5,1.5)  node[sodd,  my label={above:a}] (a) {4};
			\draw (0.0,1.5)   node[seven, my label={above:b}] (b) {0};
			\draw (1.5,1.5)   node[sodd,  my label={above:c}] (c) {2};
			\draw (-1.5,0)    node[seven, my label={below:d}] (d) {1};
			\draw (0,0)       node[sodd,  my label={below:e}] (e) {0};
			\draw (1.5,0)     node[seven, my label={below:f}] (f) {5};
			
			\draw (a) edge (d);
			\draw (b) edge (a);
			\draw (b) edge [bend left=20] (c);
			\draw (c) edge [bend left=20] (b);
			\draw (d) edge [bend left=20] (e);
			\draw (e) edge [bend left=20] (d);
			\draw (e) edge (f);
			\draw (f) edge (c);
			\draw (c) edge[in=-30,out=30,loop,looseness=4] (c);
			\draw (d) edge[in=-145,out=145,loop,looseness=6] (d);
			\end{tikzpicture}
			&
			$\Rightarrow$
			&
			\begin{tikzpicture}[baseline=(current bounding box.center)]	
			\tikzset{every edge/.append style={>=stealth,->,solid,thick,draw,text height=0.5ex,text depth=0.2ex}}
			\tikzset{every node/.append style={minimum size=6mm,draw,fill=black!10}}
			\tikzset{my label/.style args={#1:#2}{
					append after command={
						($(\tikzlastnode.center)$) coordinate [label={[label distance=5mm,black]#1:\textbf{\strut #2}}]
			}}}
			\tikzstyle{even}=[diamond,minimum size=6mm,draw]
			\tikzstyle{odd}=[regular polygon,regular polygon sides=4,minimum size=6mm,draw]
			\tikzstyle{seven}=[even,fill=white]
			\tikzstyle{sodd}=[odd,fill=white]
			
			\draw[pattern=crosshatch,pattern color=blue!25] plot [smooth cycle, tension=0.5]
			coordinates {(-0.40,2) (-0.50,0.9) (0.7,0.7) (0.9,-0.4) (2.05,-0.40) (2.30,1.90)};
			
			\draw[pattern=crosshatch,pattern color=red!25] plot [smooth cycle, tension=0.5]
			coordinates {(-2.0,1.9) (-2.4,-0.45) (0.45,-0.45) (0.5,0.5) (-0.8,0.7) (-0.95,1.85)};
			
			\draw (-1.5,1.5)  node[sodd,  my label={above:a}] (a) {4};
			\draw (0.0,1.5)   node[seven, my label={above:b}] (b) {0};
			\draw (1.5,1.5)   node[sodd,  my label={above:c}] (c) {2};
			\draw (-1.5,0)    node[seven, my label={below:d}] (d) {1};
			\draw (0,0)       node[sodd,  my label={below:e}] (e) {0};
			\draw (1.5,0)     node[seven, my label={below:f}] (f) {5};
			
			\draw (a) edge (d);
			\draw (b) edge (a);
			\draw (b) edge [bend left=20] (c);
			\draw (c) edge [bend left=20] (b);
			\draw (d) edge [bend left=20] (e);
			\draw (e) edge [bend left=20] (d);
			\draw (e) edge (f);
			\draw (f) edge (c);
			\draw (c) edge[in=-30,out=30,loop,looseness=4] (c);
			\draw (d) edge[in=-145,out=145,loop,looseness=6] (d);
			\end{tikzpicture}
		\end{tabu}
	}
	\caption{Vertex \textbf{a} is a distraction. The distraction is removed by the opponent's tangle $\{\,\textbf{d}\,\}$. We can then learn tangle $\{\,\textbf{b},\textbf{c}\,\}$.
		Similarly, vertex \textbf{f} distracts tangle $\{\,\textbf{d},\textbf{e}\,\}$.
	}
	\label{fig:distraction}
\end{figure}
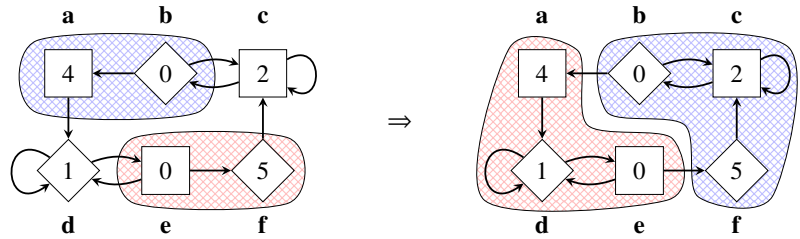

We introduce the notion of a distraction~\cite{DBLP:conf/cav/Dijk18,fpi} to further our understanding of parity game solving algorithms.
%
A \textbf{distraction} for player $\alpha$ is a vertex $v$ with an $\alpha$-priority $p$, such that if player $\alpha$ always plays to reach $v$ along paths of priorities $\leq p$, then player $\invalpha$ wins $v$ and all vertices that reach $v$.
That is, a distraction for $\alpha$ is a high value vertex $v$ with an $\alpha$-priority that player $\invalpha$ can win if player $\alpha$ always tries to visit it.
%
This occurs when player $\invalpha$ can attract $v$ to some vertex of a higher priority of player $\invalpha$,
or when player $\invalpha$ can attract $v$ to an $\invalpha$-dominion.

For attractor-based algorithms, the difficult distractions are those that are not immediately clear from the initial decomposition of the game, that is, when a distraction is attracted by the opponent \emph{via a lower tangle}.
%
%
%
%
%
Initially player $\alpha$ believes they should play to $v$, but after solving the lower subgame or finding the attracting tangle in the lower subgame, playing to $v$ is believed to be good for the opponent.
See e.g. Fig.~\ref{fig:distraction}.
Vertices \textbf{a} and \textbf{f} are distractions.
This is however not immediately clear.

We say that a distraction $v$ {distracts a tangle} $T$, with $v\notin T$, if there is some \emph{distracted} vertex $w\in T$, that is attracted by $v$ but can also avoid $v$ to play inside tangle $T$ and player $\invalpha$ cannot escape $T$ towards $v$.
In order to learn the distracted tangle, the solver first needs to identify and avoid the distraction.
As argued in~\cite{DBLP:conf/cav/Dijk18}, tangle learning identifies distractions by learning the $\invalpha$-tangle that attracts the distraction.

Zielonka's recursive algorithm and priority promotion also rely on this mechanism to avoid distractions.
In priority promotion, whenever a merged region attracts the highest vertices from $\invalpha$-regions, then these vertices are distractions.
In the recursive algorithm, the second recursive call is only performed if the opponent attracts from $A$, that is, when it discovers the attracted $p$-vertices in $A$ are distractions.
Thus, distractions make these algorithms slow, especially when lower distractions must be identified before \emph{and} after removing a higher distraction.
If there are no distractions in the game, then both the recursive algorithm and priority promotion require at most $n$ recursive calls or promotions.

Algorithms that are based on play valuations such as progress measures and strategy iteration (see~\cite{DBLP:conf/tacas/Dijk18}) rely on a fundamentally different method to deal with distractions.
They assign a higher value to vertices with $\alpha$'s priorities along the attractor paths towards vertices in $A$, and if a vertex in $A$ is a distraction, then it does not increase in value, causing these algorithms to avoid it.


\section{A Tale of Two Counters}
\label{sec:twobc}

The design of our parameterized parity game is \emph{loosely} inspired upon earlier work by Friedmann~\cite{DBLP:journals/dam/Friedmann13}, which provides an exponential lower bound for the non-oblivious strategy improvement algorithm~\cite{DBLP:conf/lpar/Fearnley10}.
That algorithm is a variation of strategy improvement that learns snares, which are related to tangles.
However, attractor-based algorithms trivially solve these games and require a different approach.

The parity game consists of two binary counters.
Each $k$th bit of the counter is a structure that contains $2^k$ many tangles.
A bit in the counter of player $\alpha$ is set by attracting a distraction for player $\invalpha$ called a ``low'' vertex via a lower $\alpha$-tangle to a ``high'' vertex.
We use two counters, one for each player, that are intertwined and progress in turns.
First the counter of player Even increases, then of player Odd, etc.
The bits are connected to the higher bits such that this tangle is no longer ``attracting'' when the higher bits of the opponent change, letting the opponent escape the tangle, thus the tangle no longer attracts the distraction and the bit is reset.
So when bit $x$ of player $\alpha$ is set, then all bits below $x$ of player $\invalpha$ are reset.
Furthermore, if lower bits \emph{of the opponent} are not set, then the tangle cannot be found because the player is distracted by these distractions.
Bit $x$ of player Odd is distracted by bits $\leq\!\!x$ of player Even, while bit $x$ of player Even is distracted by bits $<\!\!x$ of player Odd.
The result is the following example progression:
\begin{center}
	\small
	\begin{tabu} to 0.45\linewidth {X[l]X[l]X[6l]}
		\textbf{Even} & \textbf{Odd} & \textbf{Event} \\
		\textcolor{blue}{000} & \textcolor{blue}{000} & Initial state \\
		{001} & {000} & Set Even 3 \\
		\textcolor{blue}{001} & \textcolor{blue}{001} & Set Odd 3 \\
		{011} & {000} & Set Even 2 (reset Odd 3) \\
		\textcolor{blue}{010} & \textcolor{blue}{010} & Set Odd 2 (reset Even 3) \\
		{011} & {010} & Set Even 3 \\
		\textcolor{blue}{011} & \textcolor{blue}{011} & Set Odd 3 \\
		{111} & {000} & Set Even 1 (reset Odd 2, 3) \\
		\textcolor{blue}{100} & \textcolor{blue}{100} & Set Odd 1 (reset Even 2, 3) \\
		{101} & {100} & Set Even 3\\
		\textcolor{blue}{101} & \textcolor{blue}{101} & Set Odd 3 \\
		{111} & {100} & Set Even 2 (reset Odd 3) \\
		\textcolor{blue}{110} & \textcolor{blue}{110} & Set Odd 2 (reset Even 3) \\
		{111} & {110} & Set Even 3\\
		\textcolor{blue}{111} & \textcolor{blue}{111} & Set Odd 3 \\
	\end{tabu}
\end{center}
Notice that the state of the counters after every Odd bit, colored blue, tracks the progression of the counters,
while after every Even bit the counters are in an intermediate state.

\begin{figure}[tb]
	\centering
	\scalebox{0.9}{
		\begin{tikzpicture}		
		\tikzset{every edge/.append style={>=stealth,->,solid,thick,draw,text height=0.5ex,text depth=0.2ex}}
		\tikzset{every node/.append style={minimum size=6mm,draw,fill=black!10,font=\scriptsize}}
		\tikzset{my label/.style args={#1:#2}{
				append after command={
					($(\tikzlastnode.center)$) coordinate [label={[label distance=5mm,black]#1:\textbf{\strut #2}}]
		}}}
		\tikzstyle{even}=[diamond,minimum size=6mm,draw]
		\tikzstyle{odd}=[rectangle,minimum size=6mm,draw]
		\tikzstyle{seven}=[even,fill=white,inner sep=0.5mm]
		\tikzstyle{sodd}=[odd,fill=white, inner sep=0.5mm]
		\tikzstyle{sany}=[circle,fill=white, inner sep=0.5mm]
		
		\draw[thin,dashed,pattern=crosshatch dots,pattern color=blue!10] plot coordinates {(8.3,-0.4) (-2.2,-0.4) (-2.2,2.7) (8.3,2.7) (8.3,-0.4)};
		
		\draw (0,0.3)     node[ sodd, my label={above:}]  (i) {$\textbf{l}_\alpha(2)$};
		\draw (0.0,1.5)   node[ sodd, my label={above:}]  (l) {$\textbf{t}$};
		\draw (-1.5,1.5)  node[ sodd, my label={above:}]  (h) {$\textbf{h}$};
		
		\draw (1.5,1.5)   node[seven, my label={above:}] (s0) {};
		\draw (2.8,2.0)   node[ sodd, my label={above:}] (a0) {$\textbf{a}(0)$};
		\draw (3.2,1.0)   node[ sodd, my label={above:}] (b0) {$\textbf{b}(0)$};
		
		\draw (4.5,1.5)   node[seven, my label={above:}] (s1) {};
		\draw (5.8,2.0)   node[ sodd, my label={above:}] (a1) {$\textbf{a}(1)$};
		\draw (6.2,1.0)   node[ sodd, my label={above:}] (b1) {$\textbf{b}(1)$};
		
		\draw (7.5,1.5)   node[seven, my label={above:}] (s2) {$\textbf{z}$};
		
		\draw (i) edge (l);
		\draw (l) edge (h);
		\draw (l) edge (s0);
		\draw (s0) edge (a0);
		\draw (s0) edge (b0);
		\draw (a0) edge (s1);
		\draw (b0) edge (s1);
		\draw (s1) edge (a1);
		\draw (s1) edge (b1);
		\draw (a1) edge (s2);
		\draw (b1) edge (s2);
		\draw (s2) -- (7.5,0.3) -- (1,0.3) edge (l);
		
		\draw (9.2,2.0) node[draw=none,fill=none] (z3) {$\textbf{l}_\invalpha(2)$};
		\draw (9.2,1.5) node[draw=none,fill=none] (z4) {$\textbf{l}_\invalpha(3)$};
		\draw (9.2,1.0) node[draw=none,fill=none] (z5) {$\textbf{l}_\invalpha(4)$};
		\draw (s2) edge[dashed] (z3) edge[dashed] (z4) edge[dashed] (z5);
		
		\draw (2.5,3.5) node[draw=none,fill=none] (z00) {$\textbf{l}_\alpha(0)$};
		\draw (3.5,3.5) node[draw=none,fill=none] (z01) {$\textbf{l}_\invalpha(0)$};
		\draw (a0) edge[dashed] (z00);
		\draw (b0) edge[dashed] (z01);
		
		\draw (5.5,3.5) node[draw=none,fill=none] (z10) {$\textbf{l}_\alpha(1)$};
		\draw (6.5,3.5) node[draw=none,fill=none] (z11) {$\textbf{l}_\invalpha(1)$};
		\draw (a1) edge[dashed] (z10);
		\draw (b1) edge[dashed] (z11);
		
		\draw ($(h)+(-1.6,0)$) node[draw=none,fill=none] (q1) {$\textbf{l}_\alpha(1)$};
		\draw (h) edge[dashed] (q1);
		
		\draw ($(i)+(0,-1.1)$) edge[dashed] (i);
		\draw ($(i)+(-0.25,-1.05)$) edge[dashed] (i);
		\draw ($(i)+(+0.25,-1.05)$) edge[dashed] (i);
		\draw ($(i)+(-0.50,-0.95)$) edge[dashed] (i);
		\draw ($(i)+(+0.50,-0.95)$) edge[dashed] (i);
		\end{tikzpicture}
	}
	\caption{Bit $2$ of a 5-bit Two Counters game.
		We use diamonds for vertices of player $\alpha$ and boxes for vertices of player $\invalpha$.
	}
	\label{fig:abita}
\end{figure}
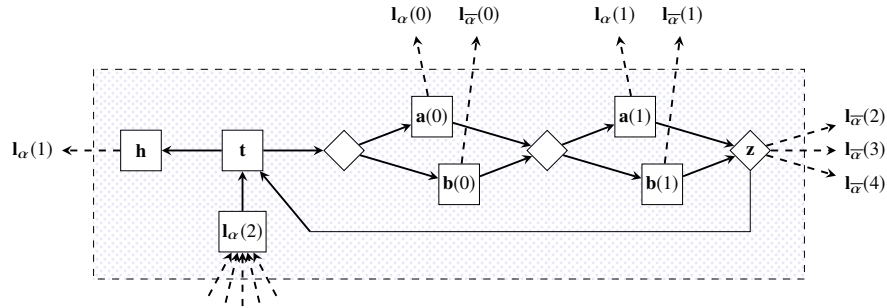

\begin{table}[tbp]
	\centering
	\begin{tabu} to 0.6\linewidth {X[3]X[3]XXX[8c]}
		\toprule
		\textbf{Player} & \textbf{Bit} & \textbf{l} & \textbf{h} & \textbf{Edge} $\textbf{z}\to\textbf{i}_\invalpha(b)$? \\
		\midrule
		Odd  & bit 0 & 8 & 15 & yes \\
		Even & bit 0 & 7 & 14 & no \\
		Odd  & bit 1 & 6 & 13 & yes \\
		Even & bit 1 & 5 & 12 & no \\
		Odd  & bit 2 & 4 & 11 & yes \\
		Even & bit 2 & 3 & 10 & no \\
		\bottomrule
	\end{tabu}
	\caption{Instantiation of a Two Counters game with $N=3$, where player Even starts.}
	\label{tbl:three}
	\vspace{-1em}
\end{table}

See Fig.~\ref{fig:abita} for bit $2$ of some player $\alpha$ in a $5$-bit counter.
The highest bit is bit $0$ and the lowest bit is bit $4$.
Vertex $\textbf{l}$ is the ``low'' vertex and has a high priority of player $\invalpha$'s parity.
This vertex is a distraction for the opponent if the bit is not set.
Vertex $\textbf{h}$ is the ``high'' vertex and has a high priority of player $\alpha$'s parity.
Vertices \textbf{h} and \textbf{l} of higher bits have higher priorities.
Vertex \textbf{h} has an edge to the input of the next higher bit, except vertex \textbf{h} of the highest bit has an edge to the input of the lowest bit.
When all bits are set, the algorithms can find the winning dominion by connecting all bits.
Vertex $\textbf{t}$ is the ``tangle'' vertex and has priority $2$ for even bits and priority $1$ for odd bits.
All other vertices have the priority $\pr(\textbf{t})-1$.



To find a tangle, player $\alpha$ must force their opponent to play from \textbf{t} to \textbf{z} such that the opponent can only escape to \textbf{h} or to higher regions of player $\alpha$.
Each tangle in the bit uses a different path from \textbf{t} to \textbf{z}, depending on the state of the higher bits.
As the two counters progress together, either $\textbf{l}_\alpha(0)$ is good for $\alpha$ and $\textbf{l}_\invalpha(0)$ is bad for $\alpha$, or vice versa.
In Fig.~\ref{fig:abita}, if bit $0$ is set and bit $1$ is unset, then the path via $\textbf{a}(0)$ and $\textbf{b}(1)$ forces the opponent to stay in the tangle, whereas they could play to their own region via $\textbf{b}(0)$ and $\textbf{a}(1)$.
Furthermore, player $\alpha$ must choose to play from \textbf{z} to \textbf{t}.
When lower $\invalpha$-bits are not set, the $\textbf{l}_\invalpha$ vertices are in a region of player $\alpha$ higher than \textbf{t} and are thus more attractive to play to, preventing player $\alpha$ from learning the tangle.
Finally, we have an asymmetry, since one player starts counting first having the lowest priority \textbf{l} and \textbf{h} vertices.
We only have the edge from vertex \textbf{z} to the matching vertex $\textbf{l}_\invalpha$ for the bits of the second player.
This forces the second player to wait until the first player sets the matching bit.

See Table~\ref{tbl:three} for an example with 3 bits,
and Fig.~\ref{fig:3bitgame} for the full 3-bit Two Counters game.
Since player Even has the lower priorities, player Even first sets bit $2$, followed by player Odd.
When all bits are set, both players have a dominion consisting of their entire counter via the blue edges.

The number of vertices for a Two Counters game with $N$ bits for both players is $3N^2+5N$ and the number of edges is $7N^2+4N$.


\newcommand{\abit}[7]{
	\def\pl{#2}
	\def\n{#3}
	\ifthenelse{#7 = 1 \or #7 = 3 \or #7 = 5 \or #7 = 7}{\def\flagone{1}}{\def\flagone{0}}
	\ifthenelse{#7 = 2 \or #7 = 3 \or #7 = 6 \or #7 = 7}{\def\flagtwo{1}}{\def\flagtwo{0}}
	\ifthenelse{#7 = 4 \or #7 = 5 \or #7 = 6 \or #7 = 7}{\def\flagthree{1}}{\def\flagthree{0}}
	

	\ifthenelse{\pl = 0}{\def\spl{seven}\def\splb{sodd}}{\def\spl{sodd}\def\splb{seven}}
	\ifthenelse{\pl = 0}{\def\coef{1.0}}{\def\coef{-1.0}}
	\pgfmathtruncatemacro{\plb}{1-\pl}

	\coordinate (boxSW) at ($(#1) + (\coef*-0.7,-0.2)$);
	\coordinate (boxNW) at ($(#1) + (\coef*-0.7,2.5)$);
	\coordinate (boxNE) at ($(#1) + (\coef*3.1+\coef*\n*3.0,2.5)$);
	\coordinate (boxSE) at ($(#1) + (\coef*3.1+\coef*\n*3.0,-0.2)$);
	
	\coordinate (outerNE) at (boxNE);
	\coordinate (outerNW) at (boxNW);
	\coordinate (outerSE) at (boxSE);
	\coordinate (outerSW) at (boxSW);
	

	\draw[thin,color=gray!40,pattern=crosshatch dots,pattern color=blue!15] (boxSW) rectangle (boxNE);
	\def\lowpr{#5}
	\pgfmathtruncatemacro{\LLL}{\lowpr-1}
	\draw ($(#1) + (0,1.5)$)       node[\splb] (h) {#6};
	\draw ($(h)  + (\coef*1.2,0)$) node[\splb] (l) {#5};
	\ifthenelse{\flagone = 1}
    {\draw ($(l)  + (0,-1.1)$)      node[\splb] (i) {#4};}
	{\draw ($(l)  + (0,-1.1)$)      node[\splb] (i) {#4};}
	\draw (i) edge (l);
	\draw (l) edge (h);
	
	\ifthenelse{\n = 0}{
		\draw ($(l) + (\coef*1.3,0)$) node[\spl] (s0) {\LLL};
		\draw (l) edge[bend left=30] (s0);
		\draw (s0) edge[bend left=30] (l);
	}{
		\pgfmathtruncatemacro{\np}{\n-1}
		\foreach \i in {0,...,\np} {
			\ifthenelse{\i = 0}
			{
			\draw ($(l)  + (\coef*1.3,0)$)      node[\spl] (s\i) {\LLL};
			}{
			\pgfmathtruncatemacro{\ip}{\i-1}
			\draw ($(s\ip) + (\coef*3,0)$)      node[\spl] (s\i) {\LLL};
			}
			\draw ($(s\i) + (\coef*1.0,0.4)$)   node[\splb] (a\i) {\LLL};
			\draw ($(s\i) + (\coef*1.8,-0.4)$)  node[\splb] (b\i) {\LLL};
		}
		\draw ($(s\np) + (\coef*3,0)$)   node[\spl] (s\n) {\LLL};
		\foreach \i in {0,...,\np} {
			\pgfmathtruncatemacro{\in}{\i+1}
			\draw (s\i) edge (a\i);
			\draw (s\i) edge (b\i);
			\draw (a\i) edge (s\in);
			\draw (b\i) edge (s\in);
			\coordinate (z\pl-\n-\pl-\i) at ($(outerNE)!(a\i)!(outerNW)$);
			\coordinate (z\pl-\n-\plb-\i) at ($(outerNE)!(b\i)!(outerNW)$);
			\draw (a\i) edge[externalb] (z\pl-\n-\pl-\i);
			\draw (b\i) edge[externalb] (z\pl-\n-\plb-\i);
		}
		\draw (l) edge (s0);
		\draw [thick] (s\n) -- ($(s\n) + (0,-0.85)$) -- ($(l) + (\coef*1.3,-0.85)$) edge (l);
	}

	\coordinate (z\pl-\n) at ($(outerNE)!(s\n)!(outerSE)$);
	\ifthenelse{\flagtwo = 1}{}{\draw (s\n) edge[externalc] (z\pl-\n);}
	\coordinate (iy\pl-\n) at ($(outerNE)!(i)!(outerSE)$);
	\ifthenelse{\flagthree = 1}{}{\draw (iy\pl-\n) edge[externalc] (i);}
	\coordinate (h\pl-\n) at ($(outerNW)!(h)!(outerSW)$);
	\draw (h) edge[externala] (h\pl-\n);
	\coordinate (ix\pl-\n) at ($(outerNW)!(i)!(outerSW)$);
	\draw (ix\pl-\n) edge[externala] (i);
	\coordinate (i\pl-\n) at ($(outerSE)!(i)!(outerSW)$);
	\ifthenelse{\flagone = 1}{}{\draw (i\pl-\n) edge[externalb] (i);}
}

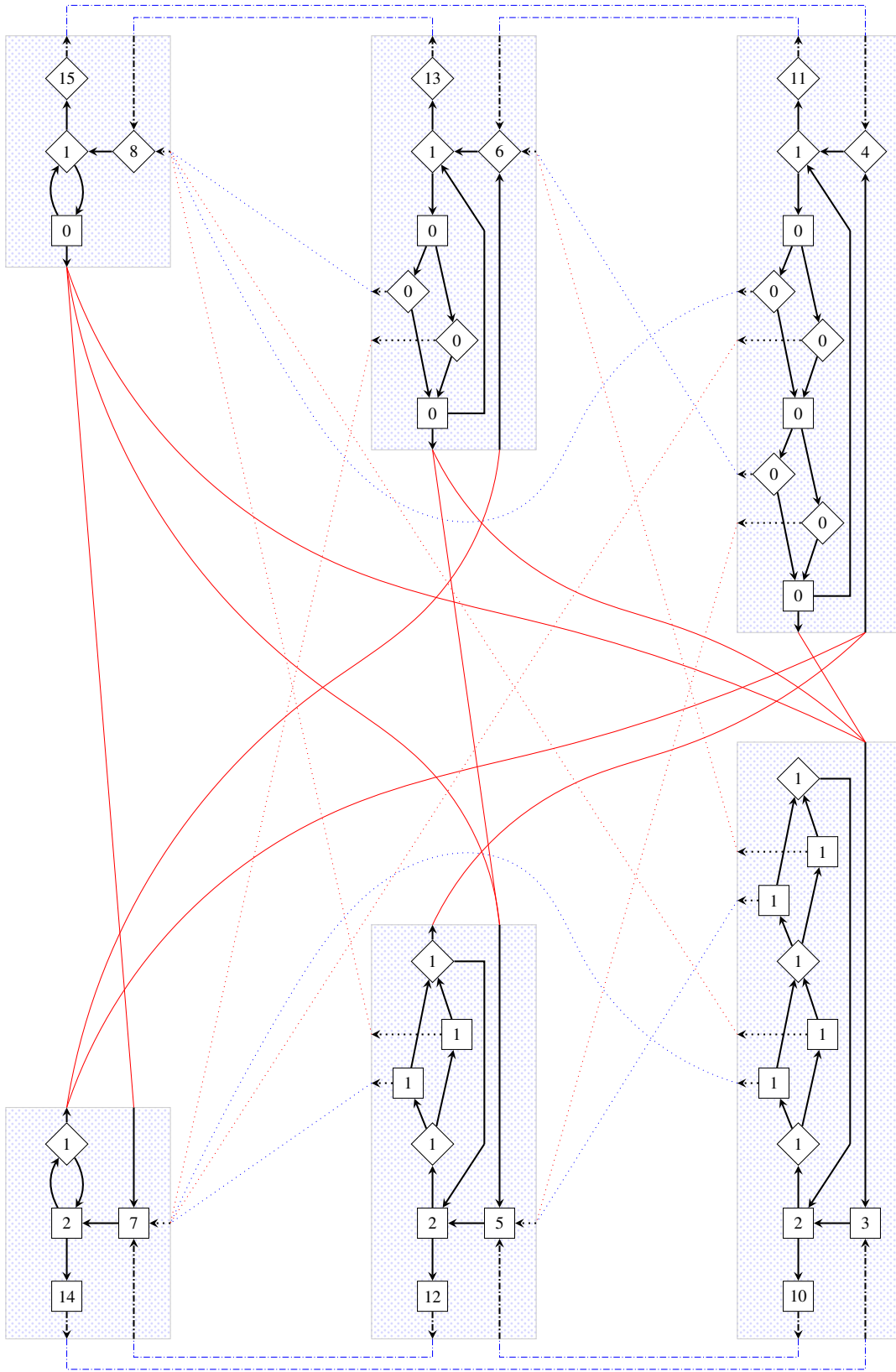
\begin{figure}[tbp]
	\centering
	\resizebox{0.9\linewidth}{!}{
		\begin{tikzpicture}[rotate=90]
		\tikzset{every edge/.append style={>=stealth,->,solid,thick,draw,text height=0.5ex,text depth=0.2ex}}
		\tikzset{every node/.append style={minimum size=7mm,minimum width=7mm,inner sep=0mm,draw,fill=white,font=\scriptsize}}
		\tikzstyle{seven}=[diamond]
		\tikzstyle{sodd}=[regular polygon, regular polygon sides=4]
		\tikzstyle{externala}=[densely dashdotted]
		\tikzstyle{externalb}=[dotted]
		\tikzstyle{externalc}=[]
		
		\abit{20,12} {1}{0}{8}{1}{15}{4}
		\abit{0,12}  {0}{0}{7}{2}{14}{0}
		\abit{20,6}  {1}{1}{6}{1}{13}{0}
		\abit{0,6}   {0}{1}{5}{2}{12}{0}
		\abit{20,0}  {1}{2}{4}{1}{11}{1}
		\abit{0,0}   {0}{2}{3}{2}{10}{3}
		
		\draw [blue,densely dashdotted] (h0-2) to ++(-0.3,0) to ++(0,4.9) to (ix0-1);
		\draw [blue,densely dashdotted] (h0-1) to ++(-0.3,0) to ++(0,4.9) to (ix0-0);
		\draw [blue,densely dashdotted] (h0-0) to (-1.2,13.5) to (-1.2,0.4) to (ix0-2);
		\draw [blue,densely dashdotted] (h1-2) to ++(0.3,0) to ++(0,4.9) to (ix1-1);
		\draw [blue,densely dashdotted] (h1-1) to ++(0.3,0) to ++(0,4.9) to (ix1-0);
		\draw [blue,densely dashdotted] (h1-0) to (21.2,13.5) to (21.2,0.4) to (ix1-2);
		
		\draw [blue,dotted] (z0-1-0-0) to (i0-0);
		\draw [blue,dotted] (z0-2-0-0) to [curve through={(5.5,5) (7,6.0) (7,8) (5,9.8)}] (i0-0);
		\draw [blue,dotted] (z0-2-0-1) to (i0-1);
		\draw [red,dotted] (z0-1-1-0) to (i1-0);
		\draw [red,dotted] (z0-2-1-0) to (i1-0);
		\draw [red,dotted] (z0-2-1-1) to (i1-1);
		
		\draw [blue,dotted] (z1-1-1-0) to (i1-0);
		\draw [blue,dotted] (z1-2-1-0) to [curve through={(14.5,5) (13,6.0) (13,8) (15,9.8)}] (i1-0);
		\draw [blue,dotted] (z1-2-1-1) to (i1-1);
		\draw [red,dotted] (z1-1-0-0) to (i0-0);
		\draw [red,dotted] (z1-2-0-0) to (i0-0);
		\draw [red,dotted] (z1-2-0-1) to (i0-1);
		
		\draw [red] (z0-0) to [curve through={(10,9.3) (11.2,7.7)}] (iy1-1);
		\draw [red] (z0-0) to [curve through={(8,9) (9,5.5)}] (iy1-2);
		\draw [red] (z0-1) to [curve through={(8.5,5) (9,3.5)}] (iy1-2);
		
		\draw [red] (z1-0) to [curve through={(10,9.3) (8.8,7.7)}] (iy0-1);
		\draw [red] (z1-0) to [curve through={(12,9) (11,5.5)}] (iy0-2);
		\draw [red] (z1-1) to [curve through={(11.5,5) (11,3.5)}] (iy0-2);
		\draw [red] (z1-0) to (iy0-0);
		\draw [red] (z1-1) to (iy0-1);
		\draw [red] (z1-2) to (iy0-2);
		\end{tikzpicture}
	}
	\caption{%
		The 3-bit Two Counters game. 
	}
	\label{fig:3bitgame}
\end{figure}



\section{Empirical evaluation}
\label{sec:data}

To assess the runtime complexity of solving $N$-bit Two Counters games with different algorithms,
we solve them with these algorithms and obtain a relevant statistic that is indicative of the runtime.
In this section, we support the claim that the games are exponential for these solvers based on this statistic.
Although this is a weaker basis than a machine-checkable complexity proof, such a proof is more difficult to produce and to understand, while the exponential behavior is clear from the empirical data.
However, in Section~\ref{sec:butwhy},
we do take a closer look at how these algorithms solve the Two Counters games,
and we sketch out an inductive proof that the algorithms require exponentially many steps for higher $N$.

To solve the games we use the implementation of the algorithms in the parity game solver \textsc{Oink}~\cite{DBLP:conf/tacas/Dijk18}, but we have also confirmed the results with implementations in \textsc{PGSolver}~\cite{DBLP:conf/atva/FriedmannL09}.
We use the following command line to obtain the results:

{\begin{center}
	\begin{minipage}{0.75\linewidth}
		\noindent
		\verb@for s in zlk pp dp rr rrdp tl atl; do@ \\
		\verb@    for i in 1 2 3 4 5 6 7 8 9 10 15 20; do@ \\
		\verb@        ./tc $i | ./oink --$s | grep "solved with"; done; done@
	\end{minipage}
\end{center}}

\begin{table}[tbp]
\begin{tabu} to \linewidth {X[0.15c]X[0.55r]|X[1r]X[r]X[r]X[1r]X[r]}
	\toprule
	\multicolumn{2}{l|}{\hfill\textbf{\#vertices}} & \textbf{ZLK} & \textbf{PP} & \textbf{DP} & \textbf{RR, RRDP} & \textbf{TL, ATL} \\
	\multicolumn{2}{l|}{\textbf{bits}\hfill} & calls & promotions & promotions & promotions & tangles \\
	\midrule
	1 & 8      & 8         & 2         & 2         & 2         & 2         \\
	2 & 22     & 21        & 9         & 7         & 6         & 6         \\
	3 & 42     & 45        & 23        & 18        & 14        & 14        \\
	4 & 68     & 91        & 52        & 43        & 30        & 30        \\
	5 & 100    & 181       & 112       & 97        & 62        & 62        \\
	6 & 138    & 359       & 235       & 210       & 126       & 126       \\
	7 & 182    & 713       & 485       & 442       & 254       & 254       \\
	8 & 232    & 1,419     & 990       & 913       & 510       & 510       \\
	9 & 288    & 2,829     & 2,006     & 1,863     & 1,022     & 1,022     \\
	10 & 350   & 5,647     & 4,045     & 3,772     & 2,046     & 2,046     \\
	15 & 750   & 180,249   & 130,961   & 122,742   & 65,534    & 65,534    \\
	20 & 1,300 & 5,767,203 & 4,194,108 & 3,931,927 & 2,097,150 & 2,097,150 \\
	\bottomrule
\end{tabu}
\vspace{-1em}
\caption{Solving the Two Counters games with different algorithms. We report the relevant statistic that represents the runtime for each algorithm.}
\label{tbl:difficultystatistics}
\end{table}

\noindent
In Table~\ref{tbl:difficultystatistics},
we report the following statistic for each solver:
\begin{itemize}[nosep]
	\item the number of recursive calls for Zielonka's recursive algorithm (ZLK)
	\item the number of promotions for priority promotion (PP, DP, RR, RRDP)
	\item the number of tangles for (alternating) tangle learning (TL, ATL)
\end{itemize}
Table~\ref{tbl:difficultystatistics} shows that the relevant statistic doubles with each higher number of bits, which supports that the runtime is exponential in the number of bits.
The number of non-dominion tangles for ATL and TL is the same as the number of promotions for RR and RRDP,
namely
$2\times(2^N-1)$ tangles
to set all bits.
Since the number of vertices and edges increases quadratically,
we conclude that the Two Counters games provide an exponential lower bound to these algorithms, namely $\Omega(2^{\sqrt{n}})$.
We also investigated whether inflation and compression~\cite{DBLP:conf/tacas/Dijk18} had any effect and this was not the case.

In Section~\ref{sec:butwhy},
we take a closer look at how these algorithms solve the Two Counters games,
and sketch an inductive proof that the algorithms require exponentially many steps for higher $N$.

\section{Analysis}
\label{sec:butwhy}

We study how the three algorithms solve a Two Counters game by looking at the timeline of events obtained from the trace output of the solver.
We then \emph{sketch a proof} arguing that any $N$-bit Two Counters game requires $2^N$ many steps for the three algorithms. 




\subsection{Recursive algorithm}

As argued in Section~\ref{sec:distractions},
to make the recursive algorithm take exponential time,
we need distractions such that after a ``higher'' distraction is attracted, all ``lower'' distractions must be attracted again.
First, we use the trace output of \textsc{Oink} to obtain a timeline of identified distractions using the command line \verb+./tc 5 | ./oink --zlk -t -t | grep "distraction" | grep "Odd"+:
\begin{center}
4, 3, 4, 2, 4, 3, 4, 1, 4, 3, 4, 2, 4, 3, 4, 0, 4, 3, 4, 2, 4, 3, 4, 1, 4, 3, 4, 2, 4, 3, 4
\end{center}
We can also draw this timeline as a tree, which is an abstraction of the recursive call graph:
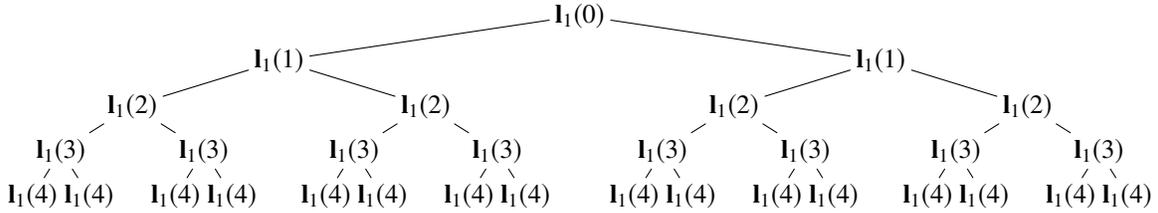
\begin{center}\small
	\begin{tikzpicture}[
	distree/.style = {align=center, inner sep=2pt, text centered, draw=none, minimum width=0.5em},
	level 1/.style={sibling distance=8cm},
	level 2/.style={sibling distance=3.9cm}, 
	level 3/.style={sibling distance=1.9cm}, 
	level 4/.style={sibling distance=0.75cm},
	level/.style={level distance = 0.6cm},
	every node/.style={distree}] 
	\node {$\textbf{l}_1(0)$}
		child{ node {$\textbf{l}_1(1)$} 
			child{ node {$\textbf{l}_1(2)$} 
				child{ node {$\textbf{l}_1(3)$} 
					child{ node {$\textbf{l}_1(4)$} }
					child{ node {$\textbf{l}_1(4)$} } }
				child{ node {$\textbf{l}_1(3)$} 
					child{ node {$\textbf{l}_1(4)$} }
					child{ node {$\textbf{l}_1(4)$} } } }
			child{ node {$\textbf{l}_1(2)$} 
				child{ node {$\textbf{l}_1(3)$} 
					child{ node {$\textbf{l}_1(4)$} }
					child{ node {$\textbf{l}_1(4)$} } }
				child{ node {$\textbf{l}_1(3)$} 
					child{ node {$\textbf{l}_1(4)$} }
					child{ node {$\textbf{l}_1(4)$} } } } }
		child{ node {$\textbf{l}_1(1)$} 
			child{ node {$\textbf{l}_1(2)$} 
				child{ node {$\textbf{l}_1(3)$} 
					child{ node {$\textbf{l}_1(4)$} }
					child{ node {$\textbf{l}_1(4)$} } }
				child{ node {$\textbf{l}_1(3)$} 
					child{ node {$\textbf{l}_1(4)$} }
					child{ node {$\textbf{l}_1(4)$} } } }
			child{ node {$\textbf{l}_1(2)$} 
				child{ node {$\textbf{l}_1(3)$} 
					child{ node {$\textbf{l}_1(4)$} }
					child{ node {$\textbf{l}_1(4)$} } }
				child{ node {$\textbf{l}_1(3)$} 
					child{ node {$\textbf{l}_1(4)$} }
					child{ node {$\textbf{l}_1(4)$} } } } } ;
	\end{tikzpicture}
\end{center}
From the above, it is at least clear that the algorithm
behaves as expected for the $5$-bit Two Counters game.
We now sketch a proof that the recursive algorithm requires exponential time for all Two Counters games.
\begin{lemma}
\label{lemma:rec}
The recursive algorithm follows the progression of a binary counter when solving a Two Counters game.
\end{lemma}
\begin{proof}
	We follow the progression of the \emph{odd} bits.
	In the initial recursive decomposition, each \textbf{h} vertex attracts no vertices and each \textbf{l} vertex attracts the connected vertices of lower bits and distracts connected \textbf{z} vertices.
	Since no \textbf{l} vertex is itself attracted, all distractions are unattracted and thus are the bits unset.
	
	Because of its strict recursive structure, the recursive algorithm cannot attract and identify a higher distraction before all lower distractions.

	After removing an \emph{even} distraction, Even attracts the entire subgame except the tangle of the corresponding bit of Odd. The reason for this is that the even bit \textbf{l} is now good-for-Even, while odd bit \textbf{l} is not yet good-for-Odd, so all lower bits are trivially attracted to the even region.
	For example, in Fig.~\ref{fig:3bitgame}, after attracting the distraction with priority $7$ (in bit $0$ of Even), vertices $7$ and $8$ are good for Even, and all lower vertices are now attracted to the even region, except the tangle of the odd bit.
	Hence, after removing the corresponding \emph{odd} distraction, the set $B$ (line~6 of Alg.~\ref{alg:zielonka}) consists exactly of the odd tangle, the attracted distraction and the directly connected lower vertices.
	That is, all the lower distractions are in the subgame that is recomputed and are distracting again.
	Thus, setting higher bits resets the lower bits.
	
	We have established that the counter starts with all bits unset, that setting higher bits resets all lower bits, and that higher bits are not set before the lower bits are set.
	Therefore the recursive algorithm follows the progression of a binary counter.
\end{proof}
Notice that in the above proof, because of the strict recursive nature,
we do not need to distract \textbf{z} to make the recursive algorithm wait until all lower bits are set before setting a higher bit.
%
In fact, removing the edges from \textbf{z} vertices to lower distractions results in even more recursive calls.
Priority promotion and tangle learning 
are not bound to a strict order like the recursive algorithm.

\subsection{Priority promotion}

In priority promotion, a bit is set when the region of tangle vertex \textbf{t} is closed,
promotes to a higher region and attracts the distraction \textbf{l}.
When solving the 3-bit TC game with the delayed promotion policy, the command line \verb+./tc 3 | ./oink --dp -t -t | egrep "promoted|delayed"+ gives us the following sequence of promotions:

\smallskip
\noindent
\begin{tabu} to \textwidth {X[0.3]X[1.2]X[1.2]X[1.1c]X[1c]X[1.5c]X[1.5c]}
\toprule
& \textbf{Promotion} & \textbf{Bit} & \multicolumn{2}{c}{\textbf{Counter state (DP)}} & \textbf{Delayed} & \textbf{Recovered} \\
   &          &         & even   & odd    &   {\small (in DP/RRDP)}          &  {\small (in RR/RRDP)}   \\
\midrule
1  & 2 to 6   & \textcolor{blue}{Even 2}  & 001     & 000    &             &     \\
2  & 1 to 5   & \textcolor{blue}{Odd 2}   & 001     & 001    &             &     \\
3  & 2 to 8   & \textcolor{blue}{Even 1}  & 011     & 000    & yes         &     \\
4  & 1 to 7   & \textcolor{blue}{Odd 1}   & 010     & 010    &             &     \\
5  & 2 to 8   & \textcolor{blue}{Even 2}  & 011     & 000    & yes         &     \\
6  & 1 to 7   & Odd 1                     & 011     & 010    &             & yes \\
7  & 1 to 7   & \textcolor{blue}{Odd 2}   & 011     & 011    &             &     \\
8  & 2 to 14  & \textcolor{blue}{Even 0}  & 111     & 000    & yes         &     \\
9  & 1 to 15  & \textcolor{blue}{Odd 0}   & 000     & 100    &             &     \\
10 & 2 to 14  & Even 0                    & 100     & 100    &             & yes \\
11 & 2 to 6   & \textcolor{blue}{Even 2}  & 101     & 100    &             &     \\
12 & 1 to 5   & \textcolor{blue}{Odd 2}   & 101     & 101    &             &     \\
13 & 2 to 14  & \textcolor{blue}{Even 1}  & 111     & 100    & yes         &     \\
14 & 1 to 15  & \textcolor{blue}{Odd 1}   & 000     & 110    &             &     \\
15 & 2 to 14  & Even 0                    & 100     & 110    &             & yes \\
16 & 2 to 14  & Even 1                    & 110     & 110    &             & yes \\
17 & 2 to 14  & \textcolor{blue}{Even 2}  & 111     & 110    &             &     \\
18 & 1 to 15  & \textcolor{blue}{Odd 2}   & 111     & 111    &             &     \\
\bottomrule
\end{tabu}

\smallskip
We record the state of the two counters after each promotion with the DP solver.
The ``recovered'' promotions do not occur in the RR and RRDP algorithms, as their regions are recovered.
All delayed promotions are immediately applied,
because after delaying the promotion, no normal promotion is available.
Thus the delayed promotion policy is not useful here.
Region recovery is however useful, as was already clear from Table~\ref{tbl:difficultystatistics}.
Both PP and DP (based on PP+) perform additional promotions that are not necessary with region recovery.
For the 3-bit TC game, PP requires $23$ promotions, DP requires $18$ promotions, RR requires $14$ promotions and RRDP requires $14$ promotions.

\begin{lemma}
\label{lemma:pp}
Priority promotion follows the progression of a binary counter when solving a Two Counters game.
\end{lemma}
\begin{proof}
Similar to the proof of Lemma~\ref{lemma:rec}.
Initially no distractions are attracted by their opponent, so the counters start with $0$.
Whenever a higher bit is reset, the PP and PP+ policies reset all lower bits of the other player and the RR algorithm cannot recover because the earlier strategy leaves the remaining subgame.
Due to the interleaving, setting odd bits resets lower even bits and vice versa.
Therefore higher bits reset lower bits.
Notice that the region of a distraction cannot be closed.
Therefore as long as a distraction is not attracted, all vertices \textbf{z} from the higher bits are distracted and cannot be part of a tangle until the distraction is attracted.
Thus high bits wait for lower bits.
From this follows Lemma~\ref{lemma:pp}.
\end{proof}

\subsection{Tangle learning}

We already explained in Section~\ref{sec:twobc} why tangle learning follows the progression of the counters,
as the counters are designed to cause tangle learning to learn all tangles in all bits.
We see that it works as expected via the command line \verb+./tc 5 | ./oink --tl -t -t -t | grep "new tangle"+ which yields precisely the expected tangles:
%
Even 4, Odd 4, Even 3, Odd 3, Even 4, Odd 4, Even 2, Odd 2, Even 4, Odd 4, Even 3, Odd 3, Even 4, Odd 4, 
Even 1, Odd 1,
Even 4, Odd 4, Even 3, Odd 3, Even 4, \dots

\newpage

\subsection{Other algorithms}

We also considered the fixpoint algorithms~\cite{DBLP:journals/corr/BruseFL14,fpi,DBLP:conf/wia/StasioMPV16}, the small progress measures algorithm~\cite{DBLP:conf/stacs/Jurdzinski00}, and the two quasi-polynomial algorithms, ordered progress measures~\cite{DBLP:conf/stoc/CaludeJKL017,DBLP:journals/sttt/FearnleyJKSSW19} and succinct progress measures~\cite{DBLP:conf/lics/JurdzinskiL17}.

The fixpoint algorithms require exponential time as they are essentially myopic versions of the recursive algorithm, based on a one-step attractor instead of full attractors.
The other algorithms are certainly slow and appear to require lifting vertices across the entire available range, but this is not due to the design of the Two Counters games.
Rather, these algorithms tend to be extremely slow even for trivial games like winning self-loops.
The exception is in fact strategy iteration~\cite{DBLP:conf/cav/Fearnley17}, which is the only algorithm that can quickly and in polynomial time solve the Two Counters games.

The Two Counters games are difficult for attractor-based algorithms for the simple reason that player $\alpha$ prefers to play to the distractions of the opponent's counters as these have the highest value for the player.
As soon as the player decides to play from \textbf{z} to $\textbf{t}$, the tangle is learned.
Algorithms using progress measures and the strategy iteration algorithm assign a higher value to vertex \textbf{t} as soon as there is a path to \textbf{t} from \textbf{z}, because vertices with $\alpha$'s priority along the path receive a progressively higher value.
Thus the players are not distracted very long by the distractions.

\section{Discussion}
\label{sec:discussion}

Since solving parity games is known to lie in $\text{UP}\cap\text{co-UP}$, it is
widely believed to admit a polynomial solution.
Researchers have been studying this problem for several decades.
For some time, the focus was on variations of strategy iteration algorithms,
where a suitable \emph{improvement rule} could lead to a polynomial solution.
Friedmann et al. provided superpolynomial or exponential lower bounds for many of such rules~\cite{DBLP:journals/mp/AvisF17,DBLP:conf/lics/Friedmann09,DBLP:journals/lmcs/Friedmann11,DBLP:conf/ipco/Friedmann11,Friedmann2012,DBLP:journals/dam/Friedmann13,DBLP:conf/soda/FriedmannHZ11},
as well as Fearnley and Savani recently~\cite{DBLP:conf/soda/FearnleyS16}.

For a long time, the main attractor-based algorithm was the recursive algorithm by McNaughton~\cite{DBLP:journals/apal/McNaughton93} and Zielonka~\cite{DBLP:journals/tcs/Zielonka98}, for which Friedmann also showed an exponential lower bound~\cite{DBLP:journals/ita/Friedmann11} and which Gazda and Willemse~\cite{DBLP:journals/corr/GazdaW13} later improved upon.
These lower bounds resist techniques like inflation, compression and SCC decomposition. 
As the lower bound of Friedmann can be defeated by memoization, since the number of distinct subgames is polynomial, Benerecetti et al. proposed a more resilient lower bound~\cite{DBLP:conf/gandalf/BenerecettiDM17}. 
Their priority promotion algorithm~\cite{DBLP:conf/cav/BenerecettiDM16} solves \cite{DBLP:conf/gandalf/BenerecettiDM17,DBLP:journals/ita/Friedmann11,DBLP:journals/corr/GazdaW13} in polynomial time, although this requires inflation for~\cite{DBLP:journals/corr/GazdaW13}.
For priority promotion, two exponential lower bounds have been presented~\cite{DBLP:journals/corr/BenerecettiDM16,DBLP:journals/fmsd/BenerecettiDM18}.
However for variations of priority promotion, in particular for DP~\cite{DBLP:conf/gandalf/BenerecettiDM17}, no superpolynomial lower bound has been published.
Tangle learning~\cite{DBLP:conf/cav/Dijk18} solves all these lower bound examples in polynomial time.

We presented the parameterized Two Counters game
and showed that this provides an exponential lower bound for
the recursive algorithm,
for (all variations of) priority promotion, in particular with the
delayed promotion policy,
and for tangle learning.
An $N$-bit Two Counters game has size $n=3N^2+5N$ and requires at least $2\times(2^{N}-1)$ steps, thus providing a lower bound of $\Omega(2^{\sqrt{n}})$.
%
%

The Two Counters game is available online via \url{https://www.github.com/trolando/oink}.

\section*{Acknowledgements}

We thank Veronika Loitzenbauer for the helpful discussions on these and other difficult parity games.
We also thank Marcin Jurdzi\'nski, Ranko Lazi\'c, Laure Daviaud and Laurent Doyen for their hospitality and discussions.
Finally, we thank Armin Biere for generously supporting this research.

\clearpage

\bibliographystyle{eptcs}
\bibliography{lit}

\end{document}